\documentclass[twocolumn,english,aps,pra,10pt,showpacs,amsmath,amssymb,floatfix,superscriptaddress,longbibliography,eqsecnum]{revtex4-1}
\usepackage{graphicx}
\usepackage{amsthm}
\usepackage{braket}
\usepackage[usenames,dvipsnames,svgnames,table]{xcolor}
\usepackage{tcolorbox}
\usepackage[unicode=true,
            pdfusetitle,
            bookmarks=true,
            bookmarksnumbered=false,
            bookmarksopen=false,
            breaklinks=true,
            pdfborder={0 0 0},
            backref=false,
            colorlinks=true,
            hypertexnames=false]{hyperref}
\hypersetup{linkcolor=NavyBlue,urlcolor=NavyBlue,citecolor=NavyBlue}

\newtheorem{lem}{Lemma} 
\DeclareMathOperator{\tr}{Tr} 						
\renewcommand{\Re}{\operatorname{Re}}
\renewcommand{\Im}{\operatorname{Im}}
\newcommand{\mt}{\mathsf{T}} 						
\newcommand{\ec}{\mathcal{E}}						
\newcommand{\mean}{\mathfrak{M}}                    
\newcommand{\qfim}{F}						        
\newcommand{\fqfi}{\mathfrak{F}}					
\newcommand{\ensemble}{{\{p_l,\ket{\psi_l}\}}} 		
\newcommand{\dyac}[1]{\ket{#1}\!\bra{#1}} 			
\newcommand{\wm}{G}									
\newcommand{\croof}{\widehat{C}}					
\newcommand{\ie}{{i.e.}}

\begin{document}
\title{Generalized-mean {Cram\'er-Rao} Bounds for Multiparameter Quantum Metrology}

\author{Xiao-Ming Lu}
\email{luxiaoming@gmail.com}
\homepage{http://xmlu.me}
\affiliation{Department of Physics, Hangzhou Dianzi University, Hangzhou 310018, China}
\affiliation{Key Laboratory of Quantum Optics, Chinese Academy of Sciences, Shanghai 200800, China}

\author{Zhihao Ma}
\email{mazhihao@sjtu.edu.cn}
\affiliation{School of Mathematical Sciences, Shanghai Jiao Tong University, Shanghai 200240, China}

\author{Chengjie Zhang}
\email{chengjie.zhang@gmail.com}
\affiliation{School of Physical Science and Technology, Soochow University, Suzhou, 215006, China}

\begin{abstract}
In multiparameter quantum metrology, the weighted-arithmetic-mean error of estimation is often used as a scalar cost function to be minimized during design optimization.
However, other types of mean error can reveal different facets of permissible error combinations.
By defining the weighted \(f\)-mean of estimation error and quantum Fisher information, we derive various quantum Cram\'er-Rao bounds on mean error in a very general form and also give their refined versions with complex quantum Fisher information matrices. 
We show that the geometric- and harmonic-mean quantum Cram\'er-Rao bounds can help to reveal more forbidden region of estimation error for a complex signal in coherent light accompanied by thermal background than just using the ordinary arithmetic-mean version.
Moreover, we show that the \(f\)-mean quantum Fisher information can be considered as information-theoretic quantities and is useful in quantifying asymmetry and coherence as quantum resources.
\end{abstract}

\maketitle

\section{Introduction}
The random nature of quantum measurement imposes fundamental limits to estimation error of unknown parameters in quantum systems. 
To reveal these fundamental limits, a variety of lower bounds on estimation error have been developed~\cite{Helstrom1976,Holevo1982,Yuen1973,Braunstein1994,Tsang2011,Tsang2012,Lu2016,Tsang2019e}.
As the most popular error bound, the quantum Cram\'er-Rao bound (QCRB) on unbiased estimator with any quantum measurement has been widely used in quantum metrology~\cite{Giovannetti2004,Giovannetti2006,Giovannetti2011,Paris2009,Lu2015}. 
For multiparameter estimation, the QCRB is given as a matrix inequality that restricts the possible error-covariance matrix of any unbiased estimation strategy, by the inverse of quantum Fisher information (QFI) matrix~\cite{Helstrom1976,Holevo1982,Yuen1973}. 
Unfortunately, this multiparameter QCRB cannot be in general saturated~\cite{Helstrom1976,Holevo1982,Yuen1973,Ragy2016,Szczykulska2016}, meaning that there might not exist an optimal measurement simultaneously minimizing the estimation errors of all parameter of interest. 
Due to Heisenberg's uncertainty principle~\cite{Heisenberg1927}, there would be a trade-off between individual estimation errors, when the optimal measurements for different parameters are not compatible in quantum mechanics. 
Since simultaneously minimizing all estimation errors of individual parameters is generally infeasible, a scalar error is in practice demanded as a cost function for optimization design.
The weighted-arithmetic-mean error is the most commonly used mean error in many previous applications of the QCRB as well as other stronger bounds like the Holevo bound~\cite{Helstrom1976,Yuen1973,Holevo1982,Tsang2011,Gessner2018,Nichols2018,Proctor2018}.

Despite the usefulness of weighted-arithmetic-mean error, many other types of mean error exist and there are no hard-and-fast rules for which mean should be used.
For example, the product of errors, which is equivalent to geometric-mean error, has been widely adopted to formulate uncertainty relations for observables that are incompatible in quantum mechanics~\cite{Heisenberg1927,Robertson1929,Robertson1934,Ozawa2003,Erhart2012,Busch2013,Branciard2013,Lu2014}. 
In fact, different types of mean error can manifest different facets of permissible error combinations for multiparameter estimation. 
This motives us to generalize the arithmetic-mean estimation errors and study their fundamental limits imposed by the random nature of quantum measurement.
To do this, we first define the \(f\)-mean estimation errors, which includes the ordinary arithmetic-mean error, geometric-mean error, and harmonic-mean error as special cases. 
We show that each \(f\)-mean estimation error is bounded from below by an \(f\)-mean version of QCRB with a corresponding \(f\)-mean QFI\@.
We also give a refined \(f\)-mean QCRB, which is tighter than the \(f\)-mean QCRB with the complex QFI matrix defined by Yuen and Lax~\cite{Yuen1973}.
Furthermore, we show that the \(f\)-mean QFIs have monotonicity under quantum operations and thus can be considered as information-theoretic quantities.
We demonstrate that they are useful in quantum resource theory, e.g., in quantifying asymmetry~\cite{Marvian2014} and coherence~\cite{Aberg,Baumgratz2014,Yuan2015,Winter2016,Yu2016,Streltsov2017,Hu2018} as quantum resources. 

\section{Quantum Cram\'er-Rao bound}\label{sec:qcrb}

Let us first introduce the QCRB and QFI, which play a pivot role in quantum parameter estimation~\cite{Helstrom1976,Holevo1982}. 
The task considered here is to estimate an unknown vector parameter \( \theta:={(\theta_1,\theta_2,\ldots,\theta_n)}^\mt \) from observations on a quantum system, where \(\mt \) denotes matrix transposition. 
The state of the quantum system is described by a density operator \(\varrho_\theta \), which depends on the value of \(\theta \).
A quantum measurement can be described by a positive-operator-valued measure \( \{M_y|M_y\geq0, \sum_y M_y=\openone \} \) with \(y\) denoting measurement outcomes and \(\openone \) the identity operator.
Denote the estimator for the \(j\)-th unknown parameter \( \theta_j \) by \( \tilde\theta_j \), which is a map from measurement outcomes \(y \) to estimates for \( \theta_j \).
The estimation error of multiple parameters can be characterized by the error-covariance matrix \( \ec \) defined by its entries
\begin{equation}\label{eq:error_covariance}
	\ec_{jk}:=\int 
	\big(\tilde \theta_j(y)-\theta_j\big)
	\big(\tilde \theta_k(y)-\theta_k\big)
	p(y|\theta)\mathrm{d}y,
\end{equation}
where \( p(y|\theta) := \tr M_y \varrho_\theta \) with \( \tr \) being trace operation is the conditional probability of obtaining a measurement outcome \(y\) for a given true value of \( \theta \).
For any unbiased estimator and any quantum measurement, the estimation error obeys the following QCRB~\cite{Helstrom1976,Yuen1973}:
\begin{equation}\label{eq:qcrb}
	\ec\geq \qfim^{-1},	
\end{equation}
where \(\qfim \) is the so-called QFI matrix~\cite{Liu2019}.
Note that the matrix inequality Eq.~\eqref{eq:qcrb} means that \( \ec - \qfim^{-1} \) is positive semi-definite.

There exist two versions of QFI matrix in the QCRB\@.
The first one is based on the symmetric logarithmic derivative (SLD) operator~\cite{Helstrom1968,Helstrom1967} and the second one is based on the right logarithmic (RLD) operator~\cite{Yuen1973}.
The SLD-based QFI matrix \( \qfim_\mathrm{S} \) is defined by \( {[\qfim_{\mathrm{S}}]}_{jk} := \Re\tr L_j L_k \varrho_\theta \), where \( \Re \) denotes the real part and \(L_j\), the SLD operator for \(\theta_j\), is a Hermitian operator satisfying \( \partial \varrho_\theta / \partial\theta_j = (L_j \varrho_\theta + \varrho_\theta L_j)/2 \).	
The RLD-based QFI matrix \(\qfim_\mathrm{R} \) is defined by 
\( {[F_\mathrm{R}]}_{jk} := \tr R_j^\dagger \varrho_\theta R_k \), where \(R_j\), the RLD operator for \(\theta_j\),  satisfies \( \partial \varrho_\theta / \partial \theta_j = \varrho_\theta R_j \).
The SLD-based QFI matrix is real symmetric while the RLD-based QFI matrix is in general Hermitian.

As the diagonal elements of \(\ec \)---estimation errors of individual parameters---might not be simultaneously minimized, one often use the weighted-mean error \(\tr\wm\ec \) as the cost function to be minimized for optimizing quantum estimation strategies, where the given weight matrix \(\wm \) is real-symmetric and positive.
It is easy to see that the weighted-mean error is bounded as \(\tr \wm \ec \geq \tr \wm \qfim^{-1}\), according to the QCRB\@. 

\section{Generalized-mean QCRB}

To generalize the mean error of estimation, we first define the weighted \(f\)-mean for a positive matrix \(X\) as  
\begin{align}\label{eq:f_mean}
	\mean_{f,\wm}(X) := f^{-1} \left( \tr \wm f(X) \right),
\end{align}
where \(f\) is a real-valued, continuous, and strictly monotonic function on the interval \( (0,+\infty) \) and the weight matrix \(G\) is  real-symmetric and positive semi-definite.
Note that whenever \(f\) is applied on a positive matrix \(X\), it means that 
\( f(X)=U\,\mathrm{diag} \{ f(x_1),f(x_2),\ldots,f(x_n)\} \, U^\dagger \), 
where \(U\) is a unitary matrix diagonalizing \(X\) as \( U^\dagger X U = \mathrm{diag}\{x_1,x_2,\ldots,x_n\} \).
Without loss of generality, we henceforth set the weight matrix \(\wm \) to be normalized, \ie{}, \(\tr \wm = 1\).
The unweighted \(f\)-mean is given by substituting \( \wm = I_n/n \) into Eq.~\eqref{eq:f_mean} and will be simply denoted by \( \mean_f(X)\), where \(I_n\) is the \(n\times n\) identity matrix. 
It is worthy to mention that \( \mean_{f,G}(\ec) \) can be written in the form of the classical weighted \(f\)-mean~\cite{Hardy1934} of the eigenvalues \(x_j\) of \( X \) as
\begin{equation}\label{eq:classical}
    \mean_{f,G}(X) = f^{-1}(\mathbb{E}[f(x_j)]),     
\end{equation}
where the expectation \(\mathbb E\) is taken regarding \(x_j\) with the probabilities \( p_j = \tr G P_j\) and \(P_j\) is the eigen-projection of \(X\) corresponding to the eigenvalue \(x_j\).

The weighted \(f\)-mean error of estimation is given by \( \mean_{f,G}(\ec) \).
For simplicity, we will use \(f\)-mean error to denote both the weighted and unweighted versions.
This \( f \)-mean error includes as special cases the arithmetic, geometric and harmonic mean error, which we will discuss in detail later. 
For the case of single parameter estimation (\( n=1 \)), \( \mean_{f,G}(\ec) \) is always reduced to the ordinary mean-square error, no matter what the function \(f\) is taken to be.  

We now derive the generalized QCRBs on the \(f\)-mean error.
Assuming that the function \(f\) is either an operator monotone or anti-monotone~\cite{Hiai2014}, we give the \(f\)-mean QCRB as follows (see Appendix~\ref{app:proof} for a detailed proof): 
\begin{align} \label{eq:fqcrb} 
    \mean_{f,G}(\ec) \geq \frac1{\mean_{f\circ\zeta,G}{(\qfim)}}
\end{align}
where \(\zeta:x\mapsto 1/x\) is the reciprocal function.
Note that a real-valued continuous function \(f\) is called operator monotone if \( f(A) \geq f(B) \) always holds, and is called operator anti-monotone if \( f(A) \leq f(B) \) always holds, whenever the two Hermitian operators  \(A\) and \(B\) satisfy \( A \geq B \geq 0 \).
Furthermore, assuming that the weighted \(f\) mean has homogeneity, \ie{}, \( \mean_{f,G}(t X) = t\,\mean_{f,G}(X) \) holds for any positive number \(t\) and any positive matrix \(X\), we can get a classical scaling \( \ec \geq 1 / \nu \,\mean_{f\circ\zeta,G}{(\qfim)} \) with respect to the number \( \nu \) of repetition of the experiment, due to the additivity of the QFI matrix.

To establish concrete \(f\)-mean QCRBs with the classical scaling \( 1 / \nu \), we need to find the operator monotone or anti-monotone functions that result in homogeneous \(f\) means.
The reader is directed to Hiai and Petz~\cite[Chapter 4]{Hiai2014} for discussions on operator monotone functions and to  Hardy, Littlewood, and P\'{o}lya~\cite[Chapter III]{Hardy1934} for discussions on the homogeneity of the \(f\) mean.
In short, the functions \(x\mapsto x^s\) for \( s \in [-1,1] \setminus \{0\} \) and  \( x\mapsto\ln(x) \) are  either operator monotone or anti-monotone (see Appendix~\ref{app:opmonotone}) and give homogeneous \(f\)-means (see Appendix~\ref{app:homogeneity});
Therefore, they are qualified for the \(f\)-mean QCRBs.
By a little abuse of notation, we adopt the convention of the generalized mean~\cite{Hardy1934} to denote  by \( \mean_{s,G} \) the weighted generalized mean for \(f:x\mapsto x^s \) with \( s\in[-1,1]\setminus \{0\} \) and specifically set \(\mean_{0,G} \) to the case of \( f:x\mapsto \ln(x) \) as \( \lim_{s \to 0} \mean_{s,G} = \mean_{0,G} \).
Also, we will use \( \mean_{s} \) for the corresponding unweighted \(f\) means.
With this notation, the generalized QCRB reads 
\begin{equation}\label{eq:fqrcb_s}
	\mean_{s,G}(\ec)\geq \nu^{-1}\mean_{-s,G}{(F)}^{-1}
\end{equation}
with \( s\in[-1,1] \).
The generalized QCRBs in Eq.~\eqref{eq:fqcrb} and Eq.~\eqref{eq:fqrcb_s} are our first main result.

We here briefly discuss the \(f\)-mean estimation errors and their properties.
Since the \(f\)-means have the classical representation as Eq.~\eqref{eq:classical}, they inherit the comparability~\cite[see Theorem 16]{Hardy1934}, namely, 
\begin{equation}\label{eq:comparability}
	\mean_{r,G} \leq \mean_{s,G} \mbox{ for }  -1 \leq r \leq s \leq 1.
\end{equation}
This comparability is a property of the \(f\)-means themselves and thus can be applied to both the \(f\)-mean estimation errors \(\mean_{s,G}(\ec)\) and the \(f\)-mean QFIs \( \mean_{s,G}(\qfim) \).
Three primary instances of the generalized means are the weighted arithmetic, geometric, and harmonic means, which are \( \mean_{1,G} \), \( \mean_{0,G} \), and \( \mean_{-1,G} \), respectively.
We list in Tab.~\ref{tab:1} the corresponding unweighted versions of \(f\)-mean estimation errors and the \(f\)-mean QFIs giving lower bounds on the estimation errors.
The difference between these three \(f\)-mean errors becomes obvious in the regions where the eigenvalues of the error-covariance matrix have a large fluctuation, e.g., one of the eigenvalues is very small while the others are considerably large, as shown in Fig.~\ref{fig:error}.

\begin{table}
\caption{\label{tab:1} Three primary instances of the unweighted \(f\)-mean errors \(\mean_s(\ec) \) and their reciprocal mean-QFIs \(\mean_{-s}(\qfim)\).
Here, \(n\) is the number of parameters to be estimated.}
\begin{ruledtabular}
\begin{tabular}{llll}
	\(s\)& \( f(x) \) & \(\mean_{s}(\ec)\) & \(\mean_{-s}(\qfim)\) \\
	\hline
	\(1\)& \( x \) & \( \tr\ec/n \) & \( n/\tr \qfim^{-1} \) \\
	\(0\)& \( \ln x \) & \( {(\det\ec)}^{1/n} \) & \( {(\det \qfim)}^{1/n} \)	\\
	\(-1\)& \( 1/x \) & \(  n/\tr\ec^{-1} \) & \( \tr \qfim/n \)
\end{tabular}
\end{ruledtabular}
\end{table}

\begin{figure}
	\includegraphics[width=1\linewidth]{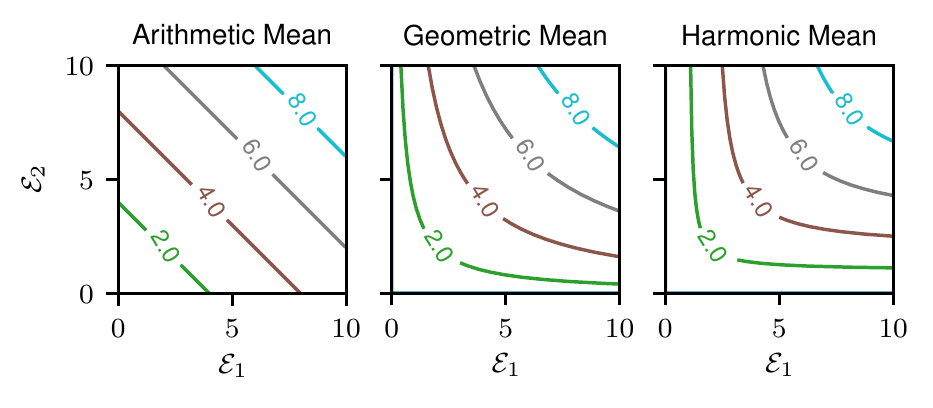}
	\caption{\label{fig:error} 
	The (unweighted) arithmetic, geometric, and harmonic means as functions of individual errors.
	Here, \(\mathcal{E}_1\) and \( \mathcal{E}_2 \) are the eigenvalues of the error-covariance matrix \(\mathcal{E}\).
	}
\end{figure}

The \(f\)-mean QCRBs in Eqs.~\eqref{eq:fqcrb} and~\eqref{eq:fqrcb_s} holds for both the SLD- and RLD-based QFI matrices.
Nevertheless, only the real part of \( f(\qfim^{-1}) \) is relevant to the \(f\)-mean QFI\@.
To see this, note that \(f(\qfim^{-1})\) is Hermitian so that its imaginary part \( \Im f(\qfim^{-1}) \) is anti-symmetric.
Since the trace of the product of a symmetric matrix and  an anti-symmetric matrix must vanish, we get \( \tr G \Im f(\qfim^{-1}) = 0 \).
As shown by Yuen and Lax~\cite{Yuen1973}, and also by Holevo~\cite{Holevo1982}, via some elaborate mathematical manipulations, the imaginary part of the RLD-based QFI matrix in fact can be used to establish a tighter bound on the weighted-arithmetic-mean estimation error than the ordination QCRB\@.
We shall generalize these result to the \(f\)-mean estimation error in what follows.  

We use Holevo's approach~\cite{Holevo1982} to refine the lower bound on the \(f\)-mean estimation error when the error-covariance matrix \( \ec \) is known to be bounded by a Hermitian matrix (\ie{}, the RLD-based QFI matrix). 
To do this, first note that for a real-symmetric matrix \(A\) and a Hermitian matrix \(B\) satisfying \( A \geq B \), it holds that~\cite[Chapter 6]{Holevo1982}
\begin{equation}\label{seq:lemma}
    \tr A \geq \tr \Re B + \lVert \Im B \rVert_1,
\end{equation}
where \(\lVert O \rVert_1 := \tr\sqrt{O^\dagger O}\) is the Schatten 1-norm of an operator \( O \).
Suppose that \(f\) is an operator monotone function.
Then, it follows from the ordinary QCRB Eq.~\eqref{eq:qcrb} and the non-negativity of the weight matrix \(G\) that 
\begin{equation}
    \sqrt G f(\ec) \sqrt G \geq \sqrt G f(\qfim^{-1}) \sqrt G.
\end{equation}
Substituting \( A = \sqrt G f(\ec) \sqrt G \) and \( B= \sqrt G f(\qfim^{-1}) \sqrt G \) into Eq.~\eqref{seq:lemma} and then applying \(f^{-1}\), we get 
\begin{align}\label{eq:refined}
	&\mean_{f,G}(\ec) \geq \mathfrak{R}_{f,G}(\qfim) := \nonumber \\
	& \ f^{-1}\Big(
        \tr G \Re f\left(\qfim^{-1}\right) 
        + \lVert \sqrt{G} \Im f(\qfim^{-1}) \sqrt{G} \rVert_1 
	\Big).
\end{align}
It is easy to see that the above inequality still holds when \(f\) is an operator anti-monotone function.

For the concrete \(f\) functions considered in this work, \ie{}, \( f: x \mapsto x^s \) with \( s\in [-1,1]\setminus \{0\} \) and \( f: x \mapsto \ln x \), the above-mentioned refined lower bound also has the classical scaling \(\nu^{-1}\) with the number \( \nu \) of repetitions of the experiments. 
To see this, note that when the experiment was repeated \( \nu \) times, the QFI matrix is given by \( \nu \qfim \) due to the additivity of the QFI matrix. 
Substituting \( f(\nu^{-1}\qfim^{-1}) = \nu^{-s} \qfim^{-s} \) for the case of \( f(x) = x^s \) and \( \ln(\nu^{-1} \qfim^{-1}) = (\ln \nu^{-1}) I + \ln(F^{-1}) \) for the case of \( f(x) = \ln x \) into Eq.~\eqref{eq:refined}, we obtain
\begin{align}\label{eq:refined2}
	\mean_{f,G}(\ec) 
	\geq \nu^{-1} \mathfrak{R}_{f,G}(\qfim).
\end{align}
The refined bound in Eqs.~\eqref{eq:refined} and~\eqref{eq:refined2} with the RLD-based QFI matrix is the second main result of this work.
This bound is tighter than the \(f\)-mean QCRB Eq.~\eqref{eq:fqcrb} with the RLD-based QFI matrix, because the term \( \lVert \sqrt{G} \Im f(\qfim^{-1}) \sqrt{G} \rVert_1 \) is nonnegative and will be reduced to Eq.~\eqref{eq:fqrcb_s} for Hermitian QFI matrices.

\section{Application to the estimation of a coherent signal}

Now, let us consider the estimation of a complex coherent signal \(\mu\) accompanied by thermal background light.
Following Helstrom~\cite{Helstrom1976}, the parametric family of density operators is given by the Glauber–Sudarshan \(P\) representation
\begin{equation}
	\varrho_\theta = \frac{1}{\pi \eta}
	\int \exp\left( - |\alpha - \mu|^2 / \eta \right) \dyac{\alpha}\mathrm{d}^2\alpha,
\end{equation} 
where \( \eta \) is the mean number of photons induced by the background, \(\ket \alpha \) is a coherent state, and \( \mu \) is a complex number.
Let us take the real and imaginary parts of \( \mu \) as the parameters \(\theta_1\) and \(\theta_2\) to be estimated, i.e., 
\begin{equation}
	\theta_1 = \Re \mu \quad \mbox{and} \quad \theta_2 = \Im \mu.
\end{equation}
The QFI matrices based on SLD and RLD have already been given in Ref.~\cite{Helstrom1976} and Ref.~\cite{Yuen1973}, respectively, that is,
\begin{align}
	\qfim_\mathrm{S} & = \frac{4}{2 \eta + 1}
	\begin{pmatrix}
		1 & 0 \\
		0 & 1
	\end{pmatrix}, \\ 
	\qfim_\mathrm{R} & = \frac{1}{\eta (\eta + 1)}
	\begin{pmatrix}
		2 \eta + 1	& -i\\
		i 		& 2 \eta + 1
	\end{pmatrix}.
\end{align}

We now calculate the \(f\)-mean versions of QCRB\@.
The unweighted \(f\)-mean QFI only depends on the eigenvalues of the QFI matrix, which are \( \{4 / ( 1 + 2 \eta ), 4 / ( 1 + 2 \eta ) \} \) for \(F_\mathrm{S}\) and \( \{ 2 / ( 1 + \eta ), 2 / \eta \} \) for \(F_\mathrm{R}\).
Since the eigenvalues of \( F_\mathrm{S} \) are the same, the unweighted \( f \)-mean SLD-based QFIs 
\begin{equation}
	\mean_s(\qfim_\mathrm{S}) = \frac{4}{1 + 2 \eta} \quad \forall s \in [-1,1].  
\end{equation}
For the RLD-based QFI, the \(f\)-mean QFI is given by
\begin{equation}
    \mean_s(\qfim_\mathrm{R}) =
    	{\left[
        	\frac12 {\left( \frac{2}{1 + \eta} \right)}^s + 
        	\frac12 {\left( \frac{2}{\eta} \right)}^s
        \right]}^{1/s}
\end{equation}
for \( s \in [-1,1] \setminus \{0\} \) and \( \mean_0(\qfim_\mathrm{R})  = 2 / {\sqrt{\eta (\eta + 1)}} \).
Due to the comparability Eq.~\eqref{eq:comparability} of the generalized means, we have \( \mean_s(\qfim_\mathrm{R}) \geq \mean_{-1}(\qfim_\mathrm{R})\).
Moreover, it is easy to show that \(\mean_{-1}(\qfim_\mathrm{R})\) equals to \( \mean_s(\qfim_\mathrm{S}) \), where the latter is in fact independent of \(s\).  
Therefore, we get \( \mean_s(\qfim_\mathrm{R}) \geq \mean_{s}(\qfim_\mathrm{S}) \) for any \(s\in[-1,1]\), meaning that the \(f\)-mean QCRB Eq.~\eqref{eq:fqrcb_s} with the SLD  gives the tighter bound for this case than that with the RLD\@.

Next, we calculate the refined \(f\)-mean QCRB with RLD\@.
Note that the inverse of the RLD-based QFI matrix has the following eigenvalue decomposition:
\begin{align}
    \qfim_\mathrm{R}^{-1}  
    &= \frac{\eta}{2} \frac{I_2 + \sigma_2}{2} + \frac{\eta + 1}{2} \frac{I_2 - \sigma_2}{2}, 
    \label{seq:spectral_decomposition}
\end{align}
where 
\(
    \sigma_2 = \begin{pmatrix}
    0 & -i \\
    i & 0
    \end{pmatrix}
\) 
is the Pauli-\(y\) matrix.
Since \( (I_2 \pm \sigma_2)/2 \) are projections, we get
\begin{equation}
    {(F_\mathrm{R}^{-1})}^s 
    = {\left( \frac{\eta}{2} \right)}^s \frac{I_2 + \sigma_2}{2}
    + {\left( \frac{\eta + 1}{2} \right)}^s \frac{I_2 - \sigma_2}{2}.
\end{equation}
By noting that the matrices \( I_2 \) and \(\sigma_2\) are purely real and imaginary, respectively, we get 
\begin{align}
    \Re {(F_\mathrm{R}^{-1})}^s 
    &= \frac12\left[
        {\left( \frac{\eta}{2} \right)}^s 
        + {\left( \frac{\eta + 1}{2} \right)}^s 
        \right] I_2, \\
    i \Im {(F_\mathrm{R}^{-1})}^s 
    &= \frac12\left[
        {\left( \frac{\eta}{2} \right)}^s 
        - {\left( \frac{\eta + 1}{2} \right)}^s 
        \right] \sigma_2.
\end{align}
According to the refined bound Eq.~\eqref{eq:refined} with \(G=I_2/2\), the unweighted \(f\)-mean error is then bounded as
\begin{align}
    &\mean_s(\ec) \geq 
    {\left[ 
        \frac12\tr{({\Re F_\mathrm{R}}^{-1})}^s 
        + \frac12 \left\lVert  {(\Im F_\mathrm{R}^{-1})}^s  \right\rVert_1
    \right]}^{1/s}  \nonumber \\
    &= {\left[ \frac12 {\left( \frac{\eta}{2} \right)}^s 
    + \frac12 {\left( \frac{\eta + 1}{2} \right)}^s 
    + \frac12 \left|
        {\left( \frac{\eta}{2} \right)}^s 
        - {\left( \frac{\eta + 1}{2} \right)}^s 
        \right|
    \right]}^{1/s} \nonumber \\
    & = \begin{cases}
        \frac{\eta}{2}, & -1 \leq s < 0\\
        \frac{\eta + 1}{2}, & 0 < s \leq 1.
    \end{cases}
\end{align}
For the case of \(s=0\), we straightforwardly calculate the refined bound as follows.  
Due to the eigenvalue decomposition Eq.~\eqref{seq:spectral_decomposition}, it follows that
\begin{align}
    \ln f(F_\mathrm{R}^{-1}) 
    &= \ln \left( \frac{\eta}{2} \right ) \frac{I_2 + \sigma_2}{2} 
    + \ln \left( \frac{\eta + 1}{2} \right) \frac{I_2 - \sigma_2}{2} \nonumber \\
    &= \frac{I_2}{2} \ln \frac{\eta(\eta+1)}{4}  + \frac{\sigma_2}{2} \ln\frac{\eta}{\eta+1},
\end{align}
from which we get \( \mean_0(\ec) \geq (\eta + 1)/2 \).

\begin{figure}
    \includegraphics[]{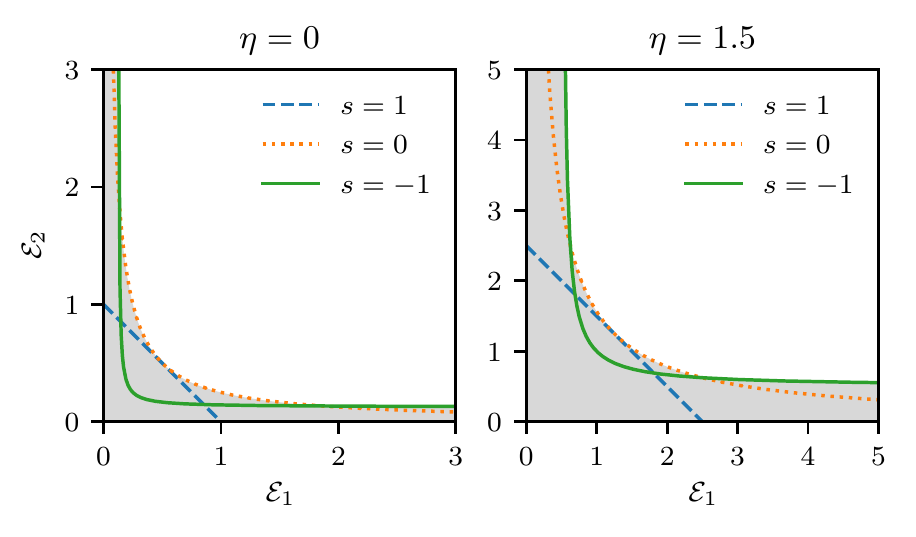}
    \caption{\label{fig:example} 
    Permissible combinations of the eigen-errors of estimating a complex coherent signal accompanied by thermal background light.
    Here \( \ec_1 \) and \( \ec_2 \) are the eigenvalues of the error-covariance matrix and \( \eta \) is the mean number of photons induced by the background.
    Only the regions above the curves are permissible by the corresponding unweighted \( f \)-mean QCRBs with different values of \( s \).
    The grey color denotes the forbidden region optimized over \(s\).
    }
\end{figure}

In this example, we can see that the refined RLD-based \(f\)-mean QCRB is tighter than the SLD-based one, when \( 0 \leq s \leq 1\), and looser when \( -1 \leq s < 0\). 
In short, the \(f\)-mean error of estimating a complex coherent signal accompanied by  thermal background light is bounded from below by
\begin{equation} \label{eq:synthesis}
    \mean_s(\ec) \geq 
    \begin{cases}
        \frac{2\eta + 1}{4}, & -1 \leq s < 0\\
        \frac{\eta + 1}{2}, & 0 \leq s \leq 1.
    \end{cases}
\end{equation}
We plot in Fig.~\ref{fig:example} the permissible combinations of the eigen-errors---the eigenvalues of the error-covariance matrix---through the \(f\)-mean QCRBs derived in this work.
It can be seen from Fig.~\ref{fig:example} that the geometric-mean and harmonic-mean QCRBs can reveal more forbidden region than the ordinary arithmetic-mean QCRB, when one of the eigen-errors is small and the other is considerably large.

To reveal more forbidden regions of error combinations, we can optimize over the \(f\) functions (i.e., \(s \in [-1,1]\) in the above example) and the weight matrices \(G\). 
In fact, the optimization with respect to \(s\) can be substantially simplified when the lower bounds on \(f\)-mean estimation error are independent of \(s\) within an interval of \(s\), due to the comparability of the \(f\)-mean error as shown in Eq.~\eqref{eq:comparability}.
Let us take the bounds given in Eq.~\eqref{eq:synthesis} as an example.
When \(-1 \leq s < 0\), \(\mean_s(\ec)\) have the common lower bound \( (2 \eta + 1) / 4 \).
From the comparability Eq.~\eqref{eq:comparability}, it can be seen that \( \mean_s(\ec) \geq \mean_{-1}(\ec) \geq (2 \eta + 1) / 4 \), implying that all error-covariance matrices \(\ec\) permissible by the harmonic-mean QCRB must be permitted by the other \(f\)-mean QCRBs with \( s \in (-1,0) \).
Analogously, it can be seen that the geometric-mean QCRB is tighter than other \(f\)-mean QCRBs with \( s\in (0,1] \). 
Therefore, in Fig.~\ref{fig:example} the union of the forbidden regions by the harmonic-mean QCRB and that by the geometric-mean QCRB has already excluded the maximal regions of eigen-error combination forbidden by Eq.~\eqref{eq:synthesis}.

\section{\texorpdfstring{\(f\)}{f}-mean QFI as information-theoretic quantity}\label{sec:fqfi}

Besides giving lower bounds on the \(f\) mean error for quantum multiparameter estimation, the \(f\)-mean QFIs can also be considered as information-theoretic quantities and may have applications on quantum information field.
For this purpose, let us treat the QFI matrix and the \(f\)-mean QFI as functions of the parametric density operators. 
We thereby denote by \( \qfim(\varrho_\theta) \) the QFI matrix of \( \varrho_\theta \) and write 
\begin{equation}\label{eq:fqfi}
	\fqfi_{f,G}(\varrho_\theta) := \mean_{f,G}(\qfim(\varrho_\theta))	
\end{equation}
as the \(f\)-mean QFI of \( \varrho_\theta \).

We here demonstrate that, like the ordinary QFI matrix, the \(f\)-mean QFIs are also monotonically non-increasing under quantum operation.
Let \( \Phi \) denote a quantum operation, which is mathematically described by a completely positive and trace-preserving linear map on  density operators.
It is known that the QFI matrix itself has the monotonicity under quantum operations~\cite[Theorem 7.34]{Hiai2014}, \ie{}, \( F(\Phi(\varrho_\theta)) \leq F(\varrho_\theta) \).
Now, suppose that \(F(\varrho_\theta)\) and \(F(\Phi(\varrho_\theta))\) are both non-degenerate.
It follows that \( F{(\varrho_\theta)}^{-1} \leq F{(\Phi{(\varrho_\theta)})}^{-1} \), for the inverse function is operator anti-monotone.
Since \(f\) is an operator monotone/anti-monotone function, it can be shown that 
\begin{align}
    f^{-1} \big(\tr G f(F{(\varrho_\theta)}^{-1}) \big) 
    \leq 
    f^{-1} \big(\tr G f(F{(\Phi{(\varrho_\theta)})}^{-1}) \big).
\end{align}
According to the definition in Eq.~\eqref{eq:fqfi}, this is equivalent to
\begin{equation}\label{eq:monotonicity}
    \fqfi_{f,G}(\Phi(\varrho_\theta)) \leq \fqfi_{f,G}(\varrho_\theta),
\end{equation}
which is the monotonicity of the \(f\)-mean QFI\@.

The monotonicity under quantum operations is an important property of information-theoretic quantities; 
This makes the \(f\) mean QFIs useful in quantum resource theory~\cite{Chitambar2019,Gour2008,Streltsov2017,Winter2016,Theurer2017,Streltsov2017a}.
The essentials of quantum resource theory are \emph{free states} and \emph{free operations}.
A resource measure is a function of quantum states that is monotonically non-increasing under free operations and vanishes for all free states~\cite{Chitambar2019}.
We show in the following that the \(f\)-mean QFIs can be used as resource measures for asymmetry~\cite{Marvian2014} and coherence~\cite{Baumgratz2014}.

Asymmetry measure was proposed by Marvian and Spekkens to quantify how much a symmetry of interest is broken for a given quantum state~\cite{Marvian2014}.
Following Ref.~\cite{Marvian2014}, the symmetry is described by a group \(G\), the free states are taken to be the symmetric states that are invariant under the action of all group elements in \( G \), and the free operations are taken to be the symmetric quantum operations \( \Phi \) that satisfy \( \Phi(U(g)\rho {U(g)}^\dagger) = U(g)\Phi(\rho){U(g)}^\dagger \) for all quantum states \( \rho \) and all \( g\in G \), where \( U(g) \) is the unitary representation of \( g \).
Now we consider the symmetry described by a Lie group whose elements are parametrized by \( \theta \in \Theta \subseteq \mathbb{R}^n \).
It can be shown that  \( \fqfi_{f,G}(U(g_\theta)\rho {U(g_\theta)}^\dagger) \) vanishes for symmetric states and monotonically non-increasing under symmetric quantum operations due to the monotonicity of the \(f\)-mean QFI\@. 
Thus, this quantity can be viewed as an asymmetry measure.
A potential application of this type is the scenario of estimating the angles of a collective SU{(2)} rotation on the spins of atom, e.g., see Ref.~\cite{Hyllus2012,Toth2012,Ma2011a}.

Besides, the \(f\)-mean QFIs can also be used to quantify quantum coherence~\cite{Aberg,Baumgratz2014,Streltsov2017,Hu2018}.
The quantum resource theory of coherence is formulated for a fixed orthonormal basis \( \{|j\rangle \} \), which will be called the reference basis. 
The free states are incoherent states whose density matrices are diagonal with the reference basis.
The definition of free operations, however, is not unique, leading to different frameworks of quantifying coherence~\cite{Streltsov2017,Hu2018}. 
In the seminal work by Baumgratz, Cramer, and Plenio (BCP)~\cite{Baumgratz2014}, the free operations are given by the incoherent Kraus operators \( K_l \), which satisfy the requirement that \( K_l \rho K_l^\dagger / \tr(K_l \rho K_l^\dagger) \) are incoherent whenever \( \rho \) is incoherent.
A nonnegative function \(C\) of \( \rho \) is said to be a coherence measure in the BCP framework, if it vanishes only when \( \rho \) is incoherent and possesses the strong monotonicity and convexity. 
The strong monotonicity means that \( C(\rho)\geq \sum_l p_l C(\rho_l) \) with \( p_l = \tr(K_l \rho K_l^\dagger) \) and \( \rho_l = K_l \rho K_l^\dagger / p_l \) holds for any set \( \{K_l\} \) of incoherent Kraus operators.
The convexity means that \( \sum_l p_l C(\rho_l)\geq C(\sum_l p_l \rho_l) \) holds for any probability distribution \( \{p_l\} \) and density operators \( \rho_l \).

To account for superpositions among all reference basis of a \(n\)-dimensional quantum system through the \(f\)-mean QFI, we consider the following parametric density operator 
\begin{equation}\label{eq:densityOperator}
	\varrho_\theta = 
    \exp\Big( i \sum_{j=1}^n \theta_j Z_j \Big) 
    \rho 
    \exp\Big( -i \sum_{j=1}^n \theta_j Z_j \Big), 
\end{equation}
where \( Z_j \)'s are a set of \(n\) Hermitian operators that are all diagonal with the reference basis and satisfy \( \tr Z_j Z_k = \delta_{jk} \).
Such \(Z_j\)'s are commuting with each other.
In fact, at most \( n - 1 \) independent parameters can be sensed into an \(n\)-dimensional quantum system by the commuting generators, for a global phase transformation does not affect the density operators.
Consequently, the QFI matrix about \(n\) parameters must be degenerated. 
To make the \(f\)-mean QFI \( \fqfi_s \) usable in quantifying coherence, it must require \( s \in (0,1] \).
In this work, we focus on the arithmetic-mean QFI \( \fqfi_1 \).

Intuitively, the superposition between the basis states is a necessary resource for estimating the unknown parameters imprinted by \(Z_j\)'s, for incoherent states are invariant under the sensing transformation \( \exp( i \sum_j \theta_j Z_j ) \).
We show in what follows that the convex roof of the arithmetic-mean QFI, as defined in the following, is a coherence measure in the BCP framework~\cite{Baumgratz2014}:
\begin{align}\label{eq:cm}
   	\croof(\rho) &:= \min_\ensemble \sum_l p_l \fqfi_1\left(e^{i\sum_j \theta_j Z_j} \dyac{\psi_l} e^{-i\sum_j \theta_j Z_j}\right) \nonumber \\
   	&= \frac{4}{n} \Big(1- \max_\ensemble \sum_l p_l\sum_{j=1}^n \lvert\braket{j|\psi_l}\rvert^4\Big),
\end{align}
where the minimization in the first line is taken over all ensemble decompositions \( \ensemble \) implementing \(\rho\) as \( \rho=\sum_l p_l \dyac{\psi_l} \) and \( \ket j \) for \(j=1,2,\ldots,n\) are the state vectors in the reference basis.
Analogous to the nomenclature for the \emph{entanglement of formation}~\cite{Horodecki2009}, we can call \( \croof(\rho) \) the \emph{arithmetic-mean QFI of formation}.

We first prove the equality in Eq.~\eqref{eq:cm}, for which the following Lemma is needed.

\begin{lem}\label{prop:invariance}
	The unweighted \(f\)-mean QFI for \(n\) parameters sensed by the commuting generators \(Z_j\)'s is the same as that sensed by another set of commuting generators given by \( Z_j' = \sum_{k=1}^n S_{jk} Z_k \), where \(S\) is an arbitrary \(n \times n\) orthogonal matrix.
\end{lem}

\begin{proof}
    It is known that the QFI matrix is transformed as \( F \mapsto S F S^\mt \) under an orthogonal transformation \( \theta \mapsto \theta' = S \theta \) of unknown parameters. 
    Since the unweighted \(f\)-mean QFIs depend only on the eigenvalues of the QFI matrix, which does not change under orthogonal transformations, we have \( \mean_{f}(S \qfim S^\mt ) = \mean_{f}(\qfim) \).
    On the other hand, the orthogonal transformation performed on the unknown vector parameter can be moved to the set of generators, as
    \begin{equation}
        \sum_{k=1}^n \theta'_k Z_k 
        = \sum_{j,k=1}^n S_{kj} \theta_j Z_k 
        = \sum_{j=1}^n \theta_j Z'_j
    \end{equation}
    with \( Z_j' = \sum_{k=1}^n S_{kj} Z_k \).
    We thus have proved the above Lemma.
\end{proof}

According to the definition of arithmetic-mean QFI, we have \( \fqfi_1(\varrho_\theta) = (1/n) \tr \qfim(\varrho_\theta) \).
For pure states, it can be shown that
\begin{align}
    \fqfi_1 &\left( 
        e^{i\sum_j \theta_j Z_j} \dyac{\psi_l} e^{-i\sum_j \theta_j Z_j} 
    \right) \nonumber \\
    & \quad = \frac4n \sum_{j=1}^n \left( \braket{\psi_l|Z_j^2|\psi_l} - \braket{\psi_l|Z_j|\psi_l}^2 \right).
    \label{seq:fqfi1}
\end{align}
Due to Lemma~\ref{prop:invariance}, \( \fqfi_1 \) is invariant under the transformation \( Z_j \to Z'_j = \sum_k S_{jk} Z_k \), so we can always choose \( Z_j = \dyac j \).
Substituting \(Z_j = \dyac j\) into Eq.~\eqref{seq:fqfi1}, we get the equality in Eq.~\eqref{eq:cm}. 

To show that \( \croof \) is a coherence measure in the BCP framework, we resort to the work by Du {\it et al}.~\cite{Du2015} and Zhu {\it et al}.~\cite{Zhu2017}:
They proved that 
\begin{equation}\label{seq:cf}
    C_f(\rho) := \min_\ensemble \sum_l p_l f(\mu(\psi_l))
\end{equation}
with \( \mu(\psi) := {(\lvert\braket{1|\psi}\rvert^2,\,\lvert\braket{2|\psi}\rvert^2,\,\ldots,\,\lvert\braket{n|\psi}\rvert^2)}^\mt \) satisfies the strong monotonicity and convexity in the BCP framework, as long as \(f\) is a real \emph{symmetric concave} function.
When we say \(f\) is symmetric, it means that \(f\) is invariant under any permutation of the elements of \( \mu \).
It is easy to see that \( \croof(\rho) \) is of the form Eq.~\eqref{seq:cf} with 
\begin{equation}\label{eq:fmu}
    f(\mu) = \frac4n \left(1-\sum_{j=1}^n \mu_j^2 \right),
\end{equation}
which is a real-valued, symmetric, and concave function
Therefore, \( \croof \) possesses the strong monotonicity and convexity in the BCP framework according to Ref.~\cite{Du2015,Zhu2017}.
Moreover, \( f(\mu) \) in Eq.~\eqref{eq:fmu} vanishes only when \( \mu \) is a sharp probability distribution such that one probability is unit and all others are zero.
As a result, \( \croof \) vanishes only when there exists an ensemble implementation of \( \rho \) such that all \( \ket{\psi_l} \)'s are in the reference basis, meaning that \( \rho \) is incoherent.
We thus have proved that \( \croof(\rho) \) is a coherence measure in the BCP framework.

Although the convex roof involved in \( \croof(\rho) \) is difficult to evaluate, we furthermore give an analytic result for 2-dimensional quantum systems (\ie{}, qubits). That is,
\begin{equation}\label{eq:qubit}
	\croof(\rho) = {(\tr\sigma_1\rho)}^2 + {(\tr\sigma_2\rho)}^2,	
\end{equation}
where \( \sigma_1 \) and \( \sigma_2 \) are the Pauli-\(x\) and -\(y\) matrices, respectively.
In order to prove Eq.~\eqref{eq:qubit}, we use Lemma~\ref{prop:invariance} to choose \( Z_1=\sigma_3/\sqrt2 \) and \( Z_2 = \openone / \sqrt2 \), where \(\sigma_3\) is the Pauli-\(z\) matrix.
We then get
\begin{align}
    \croof(\rho) 
    &= \min_\ensemble \sum_l p_l \left(
        \braket{\psi_l|\sigma_3^2|\psi_l} - \braket{\psi_l|\sigma_3|\psi_l}^2  
    \right) \nonumber \\
    &= \frac{1}{4} \qfim\left(e^{i\theta \sigma_3} \rho e^{-i\theta \sigma_3}\right),
\end{align}
where the last equality is due to the equivalence between the QFI and the convex roof of variance~\cite{Toth2013,Yu2013a}.
With the spectrum decomposition of the density operators \( \rho_\theta = \sum_\alpha \lambda_a \dyac \alpha \), it is known that the QFI can be given by~\cite{Paris2009}
\begin{align}\label{eq:explicit}
    \qfim(e^{i\theta \sigma_3} \rho e^{-i\theta \sigma_3})
    = \sum_{\alpha,\beta \mid \lambda_a+\lambda_b \neq 0} 
    \frac{2{(\lambda_\alpha-\lambda_\beta)}^2}{\lambda_\alpha+\lambda_\beta} 
    |\braket{\alpha|\sigma_3|\beta}|^2.
\end{align}
Write the density matrix in the Bloch representation
\( \rho=(\openone + \sum_{i=1}^3 r_i \sigma_i)/2 \), where \( r_i := \tr\sigma_i \rho \) is the \(i\)th component of the Bloch vector \(r\).
The eigenvalues and eigen-projections of the density matrix are \(\lambda_\pm = (1 \pm r)/2\) and
\begin{equation}
    \dyac \pm = \frac12 \left( 
        \openone \pm \sum_{i=1}^3 \frac{r_i \sigma_i}{|r|} 
    \right),
\end{equation}
respectively, where \( |r| := \sqrt{(r_1^2+r_2^2+r_3^2)} \) is the length of the Bloch vector \(r\).
Substituting these eigenvalue and eigenstates into Eq.~\eqref{eq:explicit}, we get 
\begin{align}
     \qfim(e^{i\theta \sigma_3} \rho e^{-i\theta \sigma_3}) 
    &= 4 r_1^2 + 4 r_2^2,
\end{align}
from which Eq.~\eqref{eq:qubit} immediately follows.

Moreover, we show that for any \(n\)-dimensional quantum system, \( \croof(\rho) \) is bounded as 
\begin{equation}
    0\leq \croof(\rho) \leq \frac{4 (n-1)}{n^2},
\end{equation}
and the upper bound is attained if \( \rho=\dyac{\psi}\) with \( \ket{\psi}=(1/\sqrt n)\sum_{j=1}^n \ket{j}\) up to arbitrary relative phases between \( \ket j\)'s.
The lower bound is obvious due to the non-negativity of the arithmetic-mean QFI\@.
The upper bound can be obtain from Eq.~\eqref{eq:cm} by noting that \( \sum_{j=1}^n \mu_j^2 \geq 1 / n \) for any probability distribution \( \{ \mu_j \} \) with  \( \mu_j = |\braket{j|\psi_l}|^2 \).

It is worthy to mention that Yu proposed in Ref.~\cite{Yu2017} a coherence measure that is analogous to \( \croof \) given in Eq.~\eqref{eq:cm} but using the Wigner-Yanase skew information instead of the QFI and its convex roof.
Yu also showed that reciprocal of the coherence measure thereof gives a lower bound on the harmonic-mean estimation error~\cite{Yu2017}.
Since the Wigner-Yanase skew information is not larger than QFI~\cite{Luo2004}, the reciprocal of \( \croof(\rho) \) in this work will give a tighter lower bound on the harmonic-mean estimation error than that given by the Wigner-Yanase skew information.
Besides, the convex roof of arithmetic-mean QFI has also been used by Kwon {\it et al.} in Ref.~\cite{Kwon2019} to quantify the non-classicality as a resource for quantum metrology.

\section{Conclusion}\label{sec:conclusion}
Summarizing, we have generalized the QCRB by introducing the concepts of \(f\)-mean estimation error and \(f\)-mean QFI\@.
We show that, analogous to the ordinary QCRB, the \(f\)-mean error of unbiased quantum estimation is bounded from below by the inverse of a corresponding reciprocal-\(f\)-mean QFI\@.
We have also refined the \(f\)-mean QCRB for complex QFI matrices given by the RLD approach.  
Our \(f\)-mean versions of QCRB can be used in practical application for optimization problems in multiparameter quantum metrology.
By applying our \(f\)-mean QCRBs on the scenario of estimating a complex coherent signal accompanied by background thermal light, we have demonstrated that the \(f\)-mean QCRBs can reveal more forbidden regions of error combinations than the ordinary one.
We hope the method developed in this work can help us better understand the fundamental quantum limit of multiparameter estimation. 

Moreover, we have showed that the emerged \(f\)-mean QFIs themselves can be considered as a class of information-theoretic quantities.
Like the ordinary QFI, the \(f\)-mean QFIs are monotonically non-increasing under quantum operations, which is an important property in quantum information theory.
We have demonstrated that the \(f\)-mean QFIs as well as its convex roof are useful for quantifying asymmetry and coherence in quantum resource theory. 
Considering the role of the \(f\)-mean QFI in quantum multiparameter estimation, the resource measured in such a manner can be interpreted as being valuable for the metrological purpose.

\begin{acknowledgments} 
This work is supported by 
the National Natural Science Foundation of China (Grant Nos.~61871162, 11805048, 11935012 
11571313, 
and 11734015), 
and the Natural Science Foundation of Zhejiang Province, China (Grant No.~LY18A050003). 
\end{acknowledgments}

\appendix

\section{Proof of the generalized QCRB}\label{app:proof}
We here prove the \(f\)-mean version of QCRB Eq.~\eqref{eq:fqcrb}.
The real-valued function \(f\) used in this work is supposed to be continuous, strictly monotonic, and either operator monotone or anti-monotone.
If \(f\) is operator monotone, it follows from the ordinary QCRB that
\begin{align}
    \ec \geq \qfim^{-1} 
    & \implies
    f(\ec) \geq f(\qfim^{-1}) \nonumber \\
    & \implies 
    \tr \wm f(\ec) \geq \tr \wm f(\qfim^{-1}), 
\end{align}
where the weight matrix \( \wm \) is real-symmetric and positive semi-definite.
Since \(f\) is continuous and strictly monotonic, its inverse function \(f^{-1}\) exists and must be monotonically increasing.
Therefore,
\begin{align}
	\mean_{f,G}(\ec) & = f^{-1}\left( \tr \wm f(\ec) \right) \nonumber \\
	& \geq f^{-1}\left( \tr \wm f(\qfim^{-1}) \right)  
	 = \frac1{\mean_{f\circ\zeta,G}(\qfim)},
\end{align}
where \( \zeta: x\mapsto 1/x \) is the reciprocal function.
On the other hand, if \(f\) is operator anti-monotone, then \(f\) and \( f^{-1} \) are both monotonically decreasing.
Therefore,
\begin{align}
    \ec &\geq \qfim^{-1} \implies
    f(\ec) \leq f(\qfim^{-1}) \nonumber \\
    & \implies 
	\tr \wm f(\ec) \leq \tr \wm f(\qfim^{-1}) \nonumber \\
	& \implies 
	f^{-1}\left( \tr \wm f(\ec) \right) \geq f^{-1}\left( \tr \wm f(\qfim^{-1}) \right).
\end{align}
We then still get the \(f\)-mean QCRB Eq.~\eqref{eq:fqcrb}. 

\section{Operator monotone function}\label{app:opmonotone}
We here give some concrete instances of operator monotone or anti-monotone functions. 
Remind that a function \( f:(a,b)\to \mathbb R \) is called operator monotone if \( A \geq B \) always implies \( f(A) \geq  f(B) \), where \( A \) and \( B \) are self-adjoint operators whose eigenvalues belongs to \( (a,b) \).
Similarly, \(f\) is called anti-monotone if \(A \geq B\) always implies \( f(A) \leq f(B )\).
The L\"{o}wner–Heinz inequality states that \( A \geq B \geq 0 \) implies \( A^s \geq B^s \) for all \( s \in (0,1]\).
Therefore, \( f:x\mapsto x^s\) with \( s \in [0,1]\) is an operator monotone on positive semi-definite matrices.
The function \( f: x \mapsto 1/x\) on \( (0,\infty) \) is operator anti-monotone, as 
\begin{align}
	& A\geq B>0
    \implies 
    B^{-1/2} A B^{-1/2} \geq \openone \nonumber \\
   	\implies &
    B^{1/2} A^{-1} B^{1/2} \leq \openone 
   	\implies
    A^{-1} \leq B^{-1}, 
\end{align}
where the second ``\( \implies \)'' can be seen by simultaneously diagonalizing \( B^{1/2} A^{-1} B^{1/2} \) and \( \openone \).
Combining the L\"ower-Heinze inequality with the operator anti-monotonicity of \( f: x \mapsto 1/x \), we can see that \( f:x\mapsto x^s\) for \( s \in [-1,0) \) are operator anti-monotones on \( (0,\infty) \).  
Besides, another important operator monotone function is the logarithm, for which the reader is directed to Ref.~\cite[Chapter 4]{Hiai2014}.

\section{Homogeneity of the generalized means}\label{app:homogeneity}

Remind that an \(f\)-mean is said to be homogeneous if 
\begin{equation}
	f^{-1}\Big(\sum_j p_j f(t \lambda_j)\Big) =tf^{-1}\Big(\sum_j p_j f(\lambda_j)\Big)
\end{equation}
holds for any \( t\in\mathbb R^+ \), \( \lambda_j \geq 0 \), and any probability distribution \( \{p_j\} \).
When \( f(x) = x^s \), it can be shown that
\begin{align}
    f^{-1}\big(\sum_j p_j f(t \lambda_j)\big)
    &= {\big( \sum_j p_j {(t \lambda_j)}^s \big)}^{1/s} \nonumber \\
    &= t {\big( \sum_j p_j \lambda_j^s \big)}^{1/s}.
\end{align}
When \( f(x) = \ln x \), we have
\begin{align}
      f^{-1}\Big(\sum_j p_j f(t \lambda_j)\Big)
    & = \exp \Big( \sum_j p_j \ln(t \lambda_j) \Big) \nonumber \\
    & = \exp \Big( \ln t \sum_j p_j + \sum_j p_j \ln\lambda_j \Big) \nonumber \\ 
    & = t \exp\Big( \sum_j p_j \ln\lambda_j \Big), 
\end{align}
where we have used \( \sum_j p_j = 1 \) in the last equality.
Therefore, we have shown that the \( f \) means for \( f:x\mapsto x^s\) and \( f:x\mapsto\ln x \) are homogeneous.

\end{document}